\documentclass{article}
\usepackage[affil-it]{authblk}
\usepackage{graphicx}
\usepackage[space]{grffile}
\usepackage{latexsym}
\usepackage{textcomp}
\usepackage{longtable}
\usepackage{multirow,booktabs}
\usepackage{amsfonts,amsmath,amssymb}
\usepackage{url}
\newif\iflatexml\latexmlfalse

\usepackage[english]{babel}

\usepackage[utf8]{inputenc}
\usepackage{amsmath}
\usepackage{amsfonts}
\usepackage{amssymb}
\usepackage{mathrsfs}
\usepackage{amsthm}

\numberwithin{equation}{section}

\newcommand{\A}{\mathcal{A} }
\newcommand{\C}{\mathcal{C} }
\newcommand{\V}{\mathcal{V} }
\newcommand{\Si}{\mathcal{S} }
\newcommand{\dss}{\mathcal{DSS} }
\newcommand{\B}{\mathcal{B} }

\newcommand{\Iss}{\texttt{Issuer} }
\newcommand{\Sub}{\texttt{Subject} }
\newcommand{\Head}{\texttt{Header} }
\newcommand{\Next}{\texttt{Next\_Subject} }
\newcommand{\BCS}{\texttt{Backward\_Cross\_Signature} }
\newcommand{\FCS}{\texttt{Forward\_Cross\_Signature} }
\newcommand{\T}{\mathsf{True} }
\newcommand{\F}{\mathsf{False} }

\newcommand{\PK}{\mathsf{PK}}
\newcommand{\SK}{\mathsf{SK}}

\newcommand{\Ver}{\mathsf{Ver}}
\newcommand{\Key}{\mathsf{KeyGen}}
\newcommand{\Sign}{\mathsf{Sign}}

\newcommand{\Fb}{ \mathbb{F}_{2}}

\newtheorem{theorem}{Theorem}[section]
\newtheorem{definition}{Definition}[section]
\newtheorem{corollary}{Corollary}[section]
\newtheorem{remark}{Remark}[section]

\begin{document}

\title{On the security of the Blockchain Bix Protocol and Certificates}


  
\author{Riccardo Longo, Federico Pintore, Giancarlo Rinaldo, Massimiliano Sala\footnote{\textbf{Mail}: riccardo.longo@unitn.it, federico.pintore@unitn.it, giancarlo.rinaldo@unitn.it, maxsalacodes@gmail.com}}
\affil{University of Trento}

\date{\today}

\maketitle

\begin{abstract}
The BIX protocol is a blockchain-based protocol that allows distribution of certificates linking a subject with his public key, hence providing a service similar to that of a PKI but without the need of a CA.
In this paper we analyze the security of the BIX protocol in a formal way, in four steps.
First, we identify formal security assumptions which are well-suited to this protocol. Second, we present some attack scenarios against the BIX protocol. Third, we provide a formal security proof that some of these attacks  are not feasible under our previously established assumptions. Finally, we show how another attack  may be carried on.%
\end{abstract}%

\section*{Introduction}

Blockchain is an emerging technology that is gaining widespread adoption to solve a myriad of problems where decentralized network computations can substitute a centralized approach.
Indeed, centralized computations, albeit efficient, are possible only if there is a Trusted Third Party (TTP) that everybody trusts and nowadays this is sometimes felt as a limitation.

The general idea behind blockchain technology is that blocks containing information are created by nodes in the network, and that these blocks are both public and cryptographically linked, so that an attacker should be unable to modify them without the users noting the tampering.
Also, the information contained in any block comes from the users and any user signs his own information cryptographically. Some examples of blockchain applications can be found in \cite{storj}, \cite{PoE}, \cite{ethereum}.

A very sensitive security aspect which is usually kept centralized is the issuing of digital certificates, which form the core of a Public Key Environment (PKI).
A certificate contains at least a cryptographic public key and it is digitally signed by a TTP.
An example is given by X.509 certificates \cite{X.509}, mostly containing  RSA public keys, which are widely used in the Internet for establishing secure transactions (e.g., an e-payment with an e-commerce site like Amazon). 
Since every user of a PKI must trust the Certification Authority (CA), which acts as a TTP, the identity of a web site is checked by verifying the CA's signature via the CA's public key.
We note that the identity checking is often performed via a hierarchy of CA's.  

In \cite{BIX} Muftic presents a system ("BIX") for the issuing of certificates that is based on blockchain technology and thus avoids the need of a TTP. 
We find BIX of interest and our paper aims at analyzing BIX from the point of view of security, as follows:
\begin{itemize}
    \item In Section~\ref{prel} we summarize some known preliminaries, especially on digital signatures and formal proofs of security, with focus on some notions of their security that we need for our results in Section~\ref{ch-len} e \ref{cert-tamp};
    \item In Section~\ref{bix} we provide a sketch of Muftic's scheme, highlighting its characteristics that are instrumental in its security.
    \item In Section~\ref{ch-len} we present our first security proof. Here we suppose that an attacker tries to trick the protocol by attaching a forged block to a preexisting chain, without interacting properly with the last user.
\item In Section~\ref{cert-tamp} we present our second security proof. Here we suppose that an attacker tries to modify an existing blockchain.
\item In Section~\ref{mid-alt} we show an attack on the BIX protocol which can alter an existing blockchain when
      two  members of the BIX community collude.
\item Finally, we draw some conclusions and suggest further work, possibly improving the BIX protocol.
\end{itemize}

\section{Preliminaries}
\label{prel}

\subsection{Formal Proofs of Security}
\label{f-proof}
In cryptography the security of a scheme relies on the hardness of a particular mathematical problem.
So, in a formal proof of security the goal is to model the possible attacks on the scheme and prove that a successful breach implies the solution of a hard, well-known mathematical problem.
Some security parameters may be chosen in such a way that the problem guaranteeing the security becomes almost impossible to solve (in a reasonable time), and thus the scheme becomes impenetrable.

In a formal proof of security of an encryption scheme there are two parties involved: a \emph{Challenger} $\mathcal{C}$ that runs the algorithms of the scheme and an \emph{Adversary} $\mathcal{A}$ that tries to break the scheme making queries to $\mathcal{C}$.
In a \emph{query} to $\mathcal{C}$, depending on the security model, $\mathcal{A}$ may request private keys, the encryption of specific plaintexts, the decryption of specific ciphertexts and so on.
The goal of $\mathcal{A}$ also depends on the security model, for example it may be to recover a key, to forge a digital signature, to invert a hash function.

The scheme is supposed secure if an \emph{Assumption} holds on the related mathematical problem.
Generally an \emph{Assumption} is that there is no \emph{polynomial-time} algorithm that solves a problem $\mathcal{P}$ with \emph{non-negligible} probability.
For example see the problem in the following subsection on hash functions.
The security proofs follows the following general path. Suppose there is an Adversary $\mathcal{A}$ that breaks the scheme with non-negligible probability $p_1$. A \emph{Simulator} $\mathcal{S}$ is built such that if $\mathcal{A}$ breaks the scheme then $\mathcal{S}$ solves $\mathcal{P}$.
So, given an instance of $\mathcal{P}$, $\mathcal{S}$ runs a challenger $\mathcal{C}$ that interacts with $\mathcal{A}$, simulating the scheme correctly with non-negligible probability $p_2$.
Thus $\mathcal{S}$ solves $\mathcal{P}$ with non-negligible probability, which is usually $p_1 p_2$, contradicting the Assumption.

To summarize, a formal proof of security is a \emph{reduction} from the problem \emph{attack the scheme} to the problem \emph{solve $\mathcal{P}$}.
Typically $\mathcal{P}$ is a well-studied problem so the assumption on its insolvability is accepted by the academic community.

\subsection{Hash Functions}

Commonly, the messages to be signed, seen as binary strings, are compressed in fixed-length binary string via a cryptographic hash function. A hash function $H$ can be idealized as function whose set of inputs is the set of all possible binary strings, denoted by $(\mathbb{F}_2)^*$, while its set of possible outputs is the set of all binary strings of given length (called \emph{digest}). Real-life hash functions have a \emph{finite}  input set, but  it is so large that can be thought of as infinite.
 For example, the hash functions used in the Bitcoin protocol are SHA256 (\cite{sha}), enjoying a digest length of 256 bits and with input string up to $2^{64}$-bit long, and RIPEMD-160, with 160-bit digests. 

Cryptographic hash functions can need several security assumptions, however for the goals of this paper the
four following definitions are sufficient.
\begin{definition}[Collision Problem for a Class of Inputs]\label{collpb}
Let $r \geq 1$. Let $h: (\mathbb{F}_2)^* \rightarrow (\mathbb{F})_2^r$ be a \emph{hash function}, and $L \subseteq (\mathbb{F})_2^l$ be a class of inputs. 
The \emph{collision problem} for $h$ and $L$ consists in finding two different inputs $m_1, m_2 \in L$, with $m_1 \neq m_2$, such that $h(m_1) = h(m_2)$.
\end{definition}

\begin{definition}[Collision Resistance of Hash Functions] \label{collres}
    Let $h$ be a \emph{hash function}.
    We say that $h$ is  \emph{collision resistant} for a class of inputs $L$ if there is no polynomial-time algorithm $\B(h, L) \rightarrow \{m_1, m_2\}$ that solves the Collision Problem \ref{collpb} for $h$ and $L$ with non-negligible probability.
    The complexity parameter is $r$.
 \end{definition}

\subsection{Digital Signatures and ECDSA}

To validate an action having a legal value we are usually requested to produce our handwritten signature.
Assuming that nobody is able to forge a signature while anybody can verify its validity, the signature is used to certify the correspondence of the identities of who is taking the action and in the name of whom the action is being taken.
In the digital word, handwritten signatures are substituted by \textit{digital signatures} that satisfy the same conditions seen above.
With the name \textit{Digital Signature Scheme} we refer to any asymmetric cryptographic scheme for producing and verifying digital signatures. For us, a Digital Signature Scheme  consists of three algorithms:
\begin{itemize}
\item \emph{Key Generation} - $\Key(\kappa) \rightarrow (\SK, \PK)$: given a security parameter $\kappa$ generates a public key $\PK$, that is published, and a secret key $\SK$.
\item \emph{Signing} - $\Sign(m, \SK) \rightarrow s$: given a message $m$ and the secret key $\SK$, computes a digital signature $s$ of $m$.
\item \emph{Verifying} - $\Ver(m, s, \PK) \rightarrow r$: given a message $m$, a signature $s$ and the public key $\PK$, it outputs the result $r \in \{\T, \F\}$ that says whether or not $s$ is a valid signature of $m$ computed by the secret key corresponding to $\PK$.
\end{itemize}
The previous algorithms may require a source of random bits to operate securely.

\indent
We will measure the security of a Digital Signature Scheme  by the difficulty of forging a signature in the following scheme (which results in an existential forgery):

\begin{definition}[Digital Signature Security Game]\label{digsiggame}
Let $\dss$ be a Digital Signature Scheme. Its security game for an adversary $\A$, proceeds as follows:
\begin{description}
\item [Setup.]
The challenger runs the $\Key$ algorithm, and gives to the adversary the public key $\PK$.
\item [Query.]
The adversary issues signature queries for some messages $m_i$, the challenger answers giving $s_i = \Sign(m_i, \SK)$.
\item [Challenge.]
The adversary is able to identify a message $m$ such that $m \neq m_i$~$\forall i$, and tries to compute $s$ such that $\Ver(m, s, \PK) = \T$. If $\A$ manages to do so, she wins.
\end{description}
\end{definition}

\begin{definition}[Security of a Digital Signature Scheme]\label{digsigsec}
A Digital Signature Scheme $\dss$ is said \emph{secure} if there is no polynomial-time algorithm $\A$ (w.r.t. $\kappa$) that wins the Digital Signature Security Game \ref{digsiggame} with non-negligible probability.
  \end{definition}

Ideally, a Digital Signature Scheme is designed in such a way that forging a signature in the scheme is equivalent to solving a hard mathematical problem. Although this equivalence is usually assumed but not proved, we say that the Digital Signature Scheme is based on that mathematical problem. Several Digital Signatures Schemes (e.g. \cite{CGC-cry-art-elgamal}), are based on the discrete logarithm problem (although other approaches exist, see e.g. \cite{CGC-cry-art-Rab79}, \cite{CGC-cry-art-RSA78}). Among them, the Elliptic Curve Digital Signature Algorithm (ECDSA), which uses elliptic curves, is widespread. 
To the sake of easy reference we recall briefly how ECDSA is designed (\cite{ECDSA}).\\
\newline
\emph{Domain Parameters}. An elliptic curve $\mathbb{E}$ defined over a finite field $\mathbb{F}_q$ is fixed together with a rational point $P \in \mathbb{E}(\mathbb{F}_q)$ having order $n$ and a cryptographic hash function $h$ \cite{CGC2-cry-art-Preneel99}. Let $\mathcal{O}$ be the point at infinity of $\mathbb{E}$.\\
\newline
\emph{Key Generation}. Any user A selects a random integer $d$ in the interval $[1,n-1]$. Then his \textit{public key} is $Q=dP$ while $d$ is his \textit{ private key}.\\
\newline
\emph{Signing}. To sign a message $m$ (a binary string), a user A with key pair $(Q,d)$ selects a random integer $k$ in the interval $[1,n-1]$ and accepts it if it is relatively prime to $n$. Then A computes $(x_1,y_1)=kP$ and converts $x_1$ in an integer $\bar{x}_1$. In the unlikely event that $r={\bar{x}}_1 \pmod n$ is 0, then he has to extract another random integer $k$. Otherwise, he proceeds to hash the message and to transform the digest $h(m)$ in a non-negative integer $e$ using the standard conversion of the binary representation. The A computes the value $s=k^{-1}(e+dr)~\pmod n$. 
In the unlikely event that $s=0$, the value of $k$ must be randomly selected again, else the pair $(r,s)$ is output as the A's signature for the message $m$.\\
\newline
\emph{Verifying}. To verify the signature $(r,s)$ of the public key $Q$ for the message $m$ a user B proceeds as follows: she checks if $r$ and $s$ are integers contained in the interval $[1,n-1]$; she computes the hash of the message $h(m)$ and converts it to a non-negative integer $e$. She then computes the point of the elliptic curve $(x,y)=(es^{-1}~\pmod n)P+(rs^{-1} \pmod n)Q$. If $Q=\mathcal{O}$, then B rejects the signature; otherwise she converts $x$ into an integer $\bar{x}$ and accepts the signature if and only if
$\quad(\bar{x} \pmod n)\;=\; r$.\\
We observe that if the signature $(r,s)$ of the message $m$ was actually produced by the private key $d$, then $s=k^{-1}(e+dr) \pmod n$ and so:
$$k=s^{-1}(e+dr)=es^{-1}+drs^{-1} \pmod n$$
and 
$$ kP \; =\; (x_1,y_1) \; = \; \left(es^{-1}+drs^{-1} \pmod n\right) P \,.$$
Clearly, if an attacker is able to solve the DLOG on $\mathbb{E}$, then she can break the corresponding ECDSA.
The converse is much less obvious. In \cite{signatures-paillier},
the authors provide convincing evidence that the unforgeability of several discrete-log based signatures cannot be equivalent to the discrete log problem in the standard model. Their impossibility proofs apply to many discrete-log -based signatures like ElGamal signatures and their extensions, DSA, ECDSA and KCDSA as well as standard generalizations of these. However, their work does not explicitly lead to actual attacks. Assuming that breaking the DLOG is the most efficient attack on ECDSA, then nowadays recommended key lengths start from $160$ bits, with $256$ being that of the signature employed in the Bitcoin blockchain.

\section{A description of BIX certificates}
\label{bix}
In this section we describe the BIX certificates,  the structure containing them, called the BIX Certification Ledger (BCL), and the BIX-protocol.  The BCL collects all the BIX certificates filling a double-linked list, in which every certificate is linked to the previous and the next. To simplify our notation we define the BCL as a "chain of certificates", $\mathrm{CC}$,  of $n$ certificates, that we may consider as a sequence:
\[
\mathrm{CC}:c_0,\ldots,c_{n-1}.
\]

\begin{table}[!htb]
\centering
\begin{tabular}{|l|l|l|}
\cline{2-2}
\multicolumn{1}{l|}{}             & \multicolumn{1}{c|}{  Header ($H_i$)}         &       \multicolumn{1}{l}{}      \\
\cline{2-2}
\multicolumn{1}{l|}{}             &  Sequence number          &       \multicolumn{1}{l}{}      \\
\multicolumn{1}{l|}{}             &  Version, Date.           &       \multicolumn{1}{l}{}      \\
\hline
   \multicolumn{1}{|c|}{Issuer ($S_{i-1}$)}           &      \multicolumn{1}{|c|}{Subject  ($S_i$)}        & \multicolumn{1}{|c|}{Next Subject ($S_{i+1}$)}        \\
\hline
    Bix ID of $S_{i-1}$           &      Bix ID of $S_i$        & Bix ID of $S_{i+1}$        \\
    Public key (PK$_{i-1}$)          & Public key (PK$_i$)         & Public key (PK$_{i+1}$) \\
\hline
   Issuer Signature           &      Subject Signature          & Next Subject Signature         \\
\hline
\hline
\multicolumn{2}{|c|}{Backward cross-signature} &   \\ 
\cline{1-2}

\multicolumn{2}{|l|}{Signature of $(H_i||h(S_{i-1})||h(S_i))$ by SK$_{i-1}$} &   \\ 
\multicolumn{2}{|l|}{Signature of $(H_i||h(S_{i-1})||h(S_i))$ by SK$_{i}$} &   \\ 
\hline
 & \multicolumn{2}{|c|}{Forward cross-signature}   \\ 
 \cline{2-3}

 & \multicolumn{2}{|l|}{Signature of $(H_i||h(S_{i})||h(S_{i+1}))$ by SK$_{i}$} \\ 
 &\multicolumn{2}{|l|}{Signature of $(H_i||h(S_{i})||h(S_{i+1}))$ by SK$_{i+1}$}   \\ 
\hline

\end{tabular}
\end{table}



\begin{remark}\label{rem:chain}
The owner of the certificate $c_i$ has a double role:
\begin{description}
\item[user] the owner certificates his identity by $c_i$;
\item[issuer] the owner provides the certificate $c_{i+1}$ to the next user.
\end{description}
In this way there is no need of a CA (Certification Authority).
\end{remark}

Let $\lambda(\mathrm{CC})=n$, that is, $\lambda$ is a function returning the length of a chain.
We will have $n$ users, each having a pair private-key/public key, that we denote with $\SK_i$ and $\PK_i$.
Certificate $c_0$ is called the {\texttt{root certificate}} and certificate $c_{n-1}$ is called the {\texttt{tail certificate}}.\\
In this paper, a certificate $c_i$ for $i=1,\ldots, n-2$ is defined by the following fields (and subfields) (which are necessary for our proofs of security), while the complete list can be found in \cite{BIX}. Root and tail certificates are described later on.
\begin{description}
 \item[Header ($H_i$)] 
  \ \\
 \textbf{Sequence number} $i$, the identification number of the certificate that is also the position with respect to certificates of other BIX members. 
%
\item[Subject ($S_i$)] 
 The {\Sub} contains two subfields that identify the $i$-th user~($S_i$).
\begin{description}
 \item[Subject BIX ID] This is the unique global identifier of the user who owns the certificate.\\
 In particular, all BIX ID's contained in the {\texttt{Subject}} fields of a valid chain are distinct. 
 \item[Public key] The cryptographic public key of the owner of the certificate~$\PK_i$.
 \end{description}
\item[Subject signature] It contains the signature over the  Subject attributes via the private key $\SK_i$ associated to $\PK_i$.

\item[Issuer ($S_{i-1}$)]
The \Iss field enjoys the same attribute structure of the \Sub field, but it identifies the
BIX member who certified $S_i$, i.e., it contains the \Sub attributes of $c_{i-1}$, which identifies $S_{i-1}$
(the previous member in the BCI).

\item[Issuer signature] This field contains the signature over the Issuer attributes created by the Issuer, that is, performed via the private key $\SK_{i-1}$ associated to $\PK_{i-1}$.

\item[Backward cross-signature] The {\BCS}  contains two signatures, one created by the Issuer $S_{i-1}$ and the other created by the Subject $S_i$, over the same message: the concatenation of the Header $H_i$,  the hash of the Issuer $h(S_{i-1})$ and the hash of
the Subject $h(S_i)$.\\
Note that this field guarantees validity of the Header and binding between the Subject and the Issuer.

\item[Next Subject ($S_{i+1}$)] The \Next field enjoys the same attribute structure of the \Sub field, but it identifies the
BIX member who is certified by $S_i$, i.e., it contains the \Sub attributes of $c_{i+1}$, which identifies $S_{i+1}$
(the next member in the BCI).

\item[Next Subject signature] This is the same field as Subject signature, except it is created by the
Issuer over the Next Subject data, that is, performed via the private key $\SK_{i+1}$ associated to $\PK_{i+1}$.

\item[Forward cross-signature] The {\FCS} contains two signatures, one created by the Subject $S_i$ and
the other created by the Next Subject $S_{i+1}$, over the same message: the concatenation of the Header $H_i$, the hash of the Subject $h(S_i)$ and the hash of the Next Subject $h(S_{i+1})$.\\
Note that this field guarantees binding between the current user
acting as an issuer and the next user (to whom the next certificate $c_{i+1}$ is issued).
\end{description}

We now describe the special certificates:
\begin{itemize}
\item The certificate $c_0$, called \textit{root certificate}, has the same structure of a standard certificate, but the \Iss field and the \Sub  field contain the same data. Indeed, the root user $S_0$ is not a normal user but rather an entity that initiates  the specific BCL. 
\item Also the certificate $c_{n-1}$ is special. Although it has the same structure of a standard certificate, some fields are not populated because the next user is still unknown: \Next, the \Next signature, the \FCS.\\ 
The last user that owns the last certificate, $c_{n-1}$ will then become the issuer for the next certificate (see \ref{rem:chain}).
\end{itemize}

We describe briefly how the protocol works. 
\begin{enumerate}
\item Certification request.
\begin{enumerate}
\item A new user, who will be $S_n$, asks for a certificate registers himself to the system (BCI). The system provides him the BIX-ID. Then, the user creates his private and public key, creates and signs his \Sub and sends both as a request to the system.
Since he does not know who is the last user, he sends his request to every user of the BCL. 
\item The owner of the tail certificate $c_{n-1}$, namely $S_{n-1}$, processes this request.
That is, she fills the \Iss field and the Issuer Signature field of $c_n$, while creating also an intermediate version of the \BCS field with her private key $\SK_{n-1}$.
That is, she signs $\big(H_n||h(S_{n-1})||h(S_{n})\big)$  (where "$||$" is concatenation of strings) and puts it into $c_n$.
\item At the same time, she updates her BIX certificate $c_{n-1}$ filling the \Next field and the Next Subject Signature field using the data of user $S_n$.  Moreover, she creates an intermediate version of the \FCS field by signing it with her private key . That is, she signs $\big(H_{n-1}||h(S_{n-1})||h(S_{n})\big)$ with $\SK_{n-1}$.
\item Now $S_{n-1}$ sends three certificates, $c_0$, $c_{n-1}$ and $c_n$ to the new user $S_n$ through the system (BCI).
Observe that the two certificates $c_{n-1}$ and $c_{n}$ are still incomplete and that $c_n$ will be the new tail certificate.
\item User $S_n$ receives these certificates, completes the counter-signature process by performing two digital signatures and adding them, respectively, to the \FCS of $c_{n-1}$ and the\\ \BCS  of $c_n$.
\item User $S_n$ requests the chain $\mathrm{CC}$ to the system and checks its integrity, by either traversing it forwards from $c_0$ to $c_{n-1}$  or backwards from $n-1$ to $0$ using $c_{n-1}$. If $\mathrm{CC}$ passes the integrity checks, he broadcasts $c_{n-1}$ and $c_{n}$.\\
At the same time, he stores $\mathrm{CC}$ locally for future use. 
\end{enumerate}
\item Certificate exchange.
\begin{enumerate}
\item When a user, $S_i$, wants to perform a secure communication/transaction with a  second user, $S_j$, and already both have their BIX-certificates, then he sends his certificate  $c_i$ to $S_j$ and requests her certificate $c_j$. Each user needs to verify the certificate of the other. We focus on $S_i$ since the two procedures are completely symmetric. \\
User $S_i$ checks  certificate $c_j$ in two steps:
\item he verifies the Subject signature, the Issuer signature and also the \BCS of her certificate $c_j$. To this goal, it is sufficient to have $\PK_{j-1}$, the public key of the issuer of $c_j$,  and $\PK_j$, the public key of user $S_j$ , which are both available in the certificate $c_j$ itself.
\item $S_i$ verifies that $c_j$ is  in $\mathrm{CC}$, which is surely available in his local storage in the case $i>j$ (see point 1.f). If $j>i$ and $S_i$ does not have a local version of $\mathrm{CC}$ containing the element $c_j$, he will update $\mathrm{CC}$ in his local storage by a strategy similar to that explained in point 1.f, until he reaches the $j$-th element.
\end{enumerate}
\end{enumerate}

\section{Chain Lengthening}
\label{ch-len}
  The first attack scenario that we consider supposes that an attacker tries to attach her certificate to a preexisting certificate chain without interacting properly with the last user of the chain. More precisely, the attacker $\A$ should not interact with the subject of the last certificate in the chain according to the BIX protocol.

  \begin{definition}[Static Chain Lengthening (SCL) Game]\label{sclgame}
    In this game an adversary $\A$ aims to add a certificate to the tail of a certificate chain $\mathrm{CC}$.\\
    It proceeds as follows:
    \begin{itemize}
        \item The challenger $\C$  builds a certificate chain $\mathrm{CC}$ according to the BIX protocol with root certificate $c_0$, using a hash function $h$ and a digital signature scheme $\dss$.
        \item $\C$ passes to $\A$ the complete chain $\mathrm{CC}$ together with $h$ and $\dss$.
        \item $\C$ builds an  honest verifier $\V$ that given a certificate $c^*$ and a certificate chain $\mathrm{CC}^*$ outputs $\T$ if the root certificate of $\mathrm{CC}^*$ is $c_0$ and $c^*$ is a valid certificate of $\mathrm{CC}^*$, $\F$ otherwise.
        \item $\A$ tries to build a \emph{forged} certificate chain $\mathrm{CC}'$, $\lambda(\mathrm{CC}')=n$, such that:
        \begin{itemize}
          \item $\mathrm{CC}'$ truncated before the last certificate $c'_{n}$ is identical to $\mathrm{CC}$ if the \Next and \FCS fields of the second-to-last certificate of $\mathrm{CC}'$ are not considered (i.e. we obtain $\mathrm{CC}'$ by adding a certificate  to $\mathrm{CC}$ and completing $c_{n-1}$ accordingly);
          \item user $S_{n-1}$ did not take part in the creation of $c_{n}'$ and so in particular he did not
                perform the \FCS of $c_{n-1}$ and the \BCS of $c_{n}'$;
          \item $\V(c', \mathrm{CC}') = \T $ where $c'$ is the last certificate of $\mathrm{CC}'$.
        \end{itemize}
    \end{itemize}
    $\A$ wins the \emph{SCL} game if she successfully lengthens $\mathrm{CC}$, i.e. if she buils a $\mathrm{CC}'$ that satisfies these last three points.
  \end{definition}

  \begin{definition}[Security against SCL]
    The BIX protocol is said {\texttt{secure against static chain lengthening}} if there is no adversary $\A$ that in polynomial time wins the SCL game \ref{sclgame} with non-negligible probability.
  \end{definition}

  \begin{theorem}\label{sclthm}
    Let $\A$ be an adversary that wins the SCL game \ref{sclgame} with probability $\epsilon$, then a simulator $\Si$ might be built that, with probability at least $\epsilon$, either solves the Collision Problem \ref{collpb}, with $L$ the set of all possible \Sub fields, or wins the Digital Signature Security game \ref{digsiggame}.
  \end{theorem}

  \begin{proof}
    Let $\dss$ be the digital signature scheme and $h$ the hash function used in the BIX protocol, and $L \subseteq \Fb^l$ be the class of all possible \Sub fields.
    We will build a simulator $\Si$ that simultaneously plays the Digital Signature Security (DSS) game \ref{digsiggame} and tries to solve an instance of the Collision Problem \ref{collpb} for $L$.
    It does so by simulating an instance of the SCL game \ref{sclgame} and exploiting $\A$.
    We will prove that if $\A$ wins the SCL game then either $\Si$ finds a solution for the Collision Problem or  $\Si$ wins the DSS game.

    $\Si$ starts with taking as input an instance $(h, L)$ of \ref{collpb} and a public key $\PK^*$ given
    by the $\dss$ challenger (i.e.,  the output of the first phase of \ref{digsiggame} for the scheme $\dss$).
    $\Si$ then proceeds to build a certificate chain $\mathrm{CC}^*$ following the BIX protocol.
    $\Si$ builds all but the last certificate normally, running the $\Key$ algorithm of the $\dss$ to choose public keys for the \Sub fields, so the corresponding secret keys are available to sign these certificates properly.
    Then let $n=\lambda(\mathrm{CC})^*\geq 2$ (i.e. the number of certificates contained in $\mathrm{CC}^*$), $c^*_0$ its root certificate and $c^*_{n-1}$ the last one.
    $\Si$ sets the \Sub of $c_{n-1}^*$, that we will denote by $S^*_{n-1}$, such that its public key is $\PK^*$, then it queries the challenger of the DSS game to obtain three valid signatures, respectively, on:
    \begin{itemize}
        \item the hash $h(S_{n-1}^*)$ of this subject,
        \item $\big(H_{n-1}^* || h(S_{n-2}^*) || h(S_{n-1}^*)\big)$ for the \BCS of $c_{n-1}^*$,
        \item $\big(H_{n-2}^* || h(S_{n-2}^*) || h(S_{n-1}^*)\big)$ for the \FCS of $c_{n-2}^*$,
    \end{itemize}
    where $H_{n-2}^*$ is the \Head of $c_{n - 2}^*$, $H_{n-1}^*$ is the \Head of $c_{n - 1}^*$, and $h(S_{n-2}^*)$ is the hash of the \Iss of $c_{n - 1}^*$, that is the \Sub of $c^*_{n-2}$.
    In this way $\Si$ completes a certificate chain $CC^*$ of length $n$, that it passes to $\A$.

    $\A$ responds with a counterfeit chain $\mathrm{CC}'$ of length $\lambda(\mathrm{CC})=n+1$.
    If $\mathrm{CC}'$ is not valid (the chains $\mathrm{CC}'$ and $\mathrm{CC}^*$ do not correspond up to the $(n-1)$-th certificate, or an integrity check fails) then $\Si$ discards this answer and gives up ($\Si$ fails).

    Otherwise, if the verifier outputs $\T$, the chain $\mathrm{CC}'$ is valid.
    Denote by $l'$ the string $\big(H_{n}' || h(S_{n-1}') || h(S_{n}')\big)$ signed in the \BCS of $c'_n$ (the last certificate of $\mathrm{CC}'$) by the private key corresponding to $\PK^*$.
    We have two cases:
    \begin{itemize}
    \item  $l'$ is equal to a message for which $\Si$ requested a signature.\\
    Because of its length, $l'$ may be equal to $l_0^*:=\big(H_{n-2}^* || h(S_{n-2}^*) || h(S_{n-1}^*)\big)$ or $l_1^*:=\big(H_{n-1}^* || h(S_{n-2}^*) || h(S_{n-1}^*)\big)$, but not to $h(S_{n-1}^*)$.
    In either case, $l' = l_0^*$ or  $l' = l_1^*$, the equality implies that $h(S_{n}') = h(S_{n-1}^*)$, but the specification of the BIX protocols supposes that different certificates have a different BIX ID in the \Sub
    (and we know that $\mathrm{CC}'$ is valid).
    So  $S'_{n-1}=S^*_{n-1} \neq S'_n$, because of the BIX ID's, but they have the same hash so $\Si$ may submit $(S^*_{n-1}, S'_n)$ as a solution to the Collision Problem.

\item $l'$ is different from all messages for which $\Si$ requested a signature.\\
    In the \BCS of $c'_n$ there is a signature $s$ of $l'$ such that $\Ver(l', s, \PK^*) = \T$ (remember that $\PK^*$ is the public key of the \Iss of $c'_n$ and that $\mathrm{CC}'$ is considered valid, so the signatures check out), so $\Si$ may submit $(l', s)$ as a winning answer of the challenge phase of the DSS game.
\end{itemize}

    So if $\Si$ does not fail it correctly solves the Collision Problem or wins the DSS game, and since $\A$ is a polynomial-time algorithm, $\Si$ is a polynomial-time algorithm too, given that the other operations performed correspond to the building of a certificate chain and this must be efficient.
    $\Si$ might fail only if the chain given by $\A$ is not valid (i.e. if $\A$ fails).
    Since the simulation of the SCL game is always correct, $\A$'s failure happens with probability $1 - \epsilon$, then the probability that $\Si$ wins is $1-(1-\epsilon)=\epsilon$.
  \end{proof}

  \begin{corollary}[SCL Security]\label{sclcor}
    If the Digital Signature Scheme is secure (see Assumption \ref{digsigsec}) and the hash function is collision resistant for the class $L$ (Assumption \ref{collres})  where $L$ is the set of all possible \Head fields, then the BIX protocol is secure against the Static Chain Lengthening.
  \end{corollary}

  \begin{proof}
    Thanks to Theorem \ref{sclthm}, given a polynomial-time adversary that wins the SCL game with non-negligible probability $\epsilon$, a polynomial-time simulator might be built that with the same probability either solves the Collision Problem \ref{collpb} or wins the Digital Signature Security game \ref{digsiggame}.
    So let $C$ be the event "solution of the Collision Problem" and $D$ be the event "victory at the DSS game".
    We have that 
    \begin{equation}
        \epsilon = P(C \vee D) \leq P(C) + P(D)
    \end{equation}
    Note  that the sum of two negligible quantities is itself negligible, so the fact that $\epsilon$ is non-negligible implies that at least one of $P(C)$ and $P(D)$ is non-negligible, and this means that Assumption \ref{collres} or Assumption \ref{digsigsec} is broken.
  \end{proof}
\ \\

  Let us call Alice the user of the second-to-last certificate in the chain and Bob the user of the last certificate. We observe that the infeasibility of this attack would guarantee also the non-repudiation property of the last certificate in the chain.
  That is, if Alice tries to repudiate Bob (with an eye to issuing another certificate), then Bob might claim his righteous place since no one can attach its certificate to the tail of the certificate chain without being a proper user.

  Note also that even though the security of the digital signature is presented in the form of \emph{existential forgery}, in the proof the freedom of the attacker in the choice of the message to be signed is limited.
  In fact it has to forge a signature of $l':=\big(H_{n}' || h(S_{n-1}') || h(S_{n}')\big)$,
  where $h(S_{n-1}')$ is given by $\Si$, and even $H_{n}'$ is not completely controlled by $\A$ (the sequence number is given, and the other fields should be given by the IM).
  So a large part of the string to be signed is beyond the control of the forger, so the challenge is not completely an existential forgery but something in between an existential and a universal forgery.

\section{Certificate Tampering}

\label{cert-tamp}
  In the second attack scenario that we consider, we suppose that only the root certificate of a certificate chain is publicly available and trustworthy, while a malicious attacker tries to corrupt a chain of certificates built upon this root, resulting in another chain that may re-distribute as a proper chain with same root but with altered information.
  Note that the security against this attack would guarantee that no external attacker can modify any certificate in the chain, including deleting or inserting a certificate in any non-ending point, as long as the root certificate is safe (no unauthorized usage), secure (cannot be broken) and public (anyone can check it).
  If the security proved in the previous section is also considered, then a certificate chain is also secure at the end point (no one can wrongfully insert herself at the end or disavow the last certificate) achieving full security from external attacks.

  \begin{definition}[Static Tampering with Subject (STS) Game]\label{stsgame}
    In this game an adversary $\A$ aims to modify information contained in the \emph{Subject} field of a certificate $c_i$ contained in a certificate chain $\mathrm{CC}$, with $1\leq i\leq n-2$, $n=\lambda(\mathrm{CC})$.
    It proceeds as follows:
    \begin{itemize}
        \item The challenger $\C$ builds a certificate chain $\mathrm{CC}$ with root certificate $c_0$, according to the BIX protocol and using a hash function $h$ and a digital signature scheme $\dss$. Let $n=\lambda(\mathrm{CC})$.
        \item $\C$ passes to $\A$ the complete chain $\mathrm{CC}$ together with $h$ and $\dss$.
        \item $\C$ builds an  honest verifier $\V$ that, given a certificate $c^*$ and a certificate chain $\mathrm{CC}^*$ outputs $\T$ if the root certificate of $\mathrm{CC}^*$ is $c_0$ and $c^*$ is a valid certificate of $\mathrm{CC}^*$, $\F$ otherwise.
        \item $\A$ selects a "target certificate" $c_i$ in $\mathrm{CC}$ and tries to build a forged certificate chain $\mathrm{CC}'$ such that:
        \begin{itemize}
          \item $\mathrm{CC}'$ has the same length as $\mathrm{CC}$ truncated after $c_i$, i.e. the last certificate $c'_i$ of $\mathrm{CC}'$ is in the same position $i$ as $c_i$ in $\mathrm{CC}$ relatively to $c_0$,
          with $\lambda(\mathrm{CC}')=i+1$.
          \item The \Sub fields of $c_i$ and $c_i'$ are different, that is, $S_i \neq S_i'$.
          \item $\V(c', \mathrm{CC}') = \T$
        \end{itemize}
    \end{itemize}
    $\A$ wins the \emph{STS} game if he achieves the last three items, i.e. he successfully builds such a $\mathrm{CC}'$.
  \end{definition}

  \begin{definition}[Security against STS]
    The BIX protocol is said \emph{secure against Static Tampering with Subject} if there is no adversary $\A$ that in polynomial time wins the STS game \ref{stsgame} with non-negligible probability.
  \end{definition}

  \begin{theorem}\label{ststhm}
    Let $\A$ be an adversary that wins the STS game \ref{stsgame} with probability $\epsilon$, then a simulator $\Si$ might be built that with probability at least $\frac{\epsilon}{n-1}$ either solves the collision problem \ref{collpb}, where $L$ is the set of all possible \Sub fields, or wins the digital signature security game \ref{digsiggame}, where $n$ is the length of the certificate chain that $\Si$ gives to $\A$.
  \end{theorem}

  \begin{proof}
    Let $\dss$ be the digital signature scheme and $h$ the hash function used in the BIX protocol, and $L \subseteq (\Fb)^l$ be the class of all possible \Sub fields.
    We will build a simulator $\Si$ that simultaneously plays the digital signature security (DSS) game \ref{digsiggame} and tries to solve an instance of the collision problem \ref{collpb} for $L$.
    It does so by simulating an instance of the STS game \ref{stsgame} and exploiting $\A$.
    We will prove that when $\A$ wins the STS game, one in $n-1$ times $\Si$ is successful.
    To be more precise, if $\Si$ does not find a solution for the collision problem then $\Si$ wins the DSS game.

    $\Si$ starts with taking as input an instance $(h, L)$ of \ref{collpb} and a public key $\PK^*$ as the output of the first phase of \ref{digsiggame} for the scheme $\dss$.\\
    $\Si$ now proceeds to build a certificate chain $\mathrm{CC}$ following the BIX protocol, as follows.
    First, $\Si$ chooses $n \geq 2$ (possibly depending on the $\A$'s requirements). Then $\Si$ selects $2 \leq k \leq n-2 n$ at random to be the index of a certificate $c_{k}$ in $\mathrm{CC}$.
    $\Si$ builds the first $k-1$ certificates normally, running the $\Key$ algorithm of the $\dss$ scheme to choose public keys for the \Sub fields, so the corresponding secret keys are available (to $\Si$) to sign these certificates properly. So, $c_0,\ldots,c_{k-3}$ are complete certificate and $c_{k-2}$ is a tail certificate.
    Then it sets the \Sub of $c_{k-1}$ such that its public key is $\PK^*$, and a header $H_{k-1}$, then it queries the challenger of the DSS game to obtain three valid signatures, respectively, on:
    \begin{itemize}
        \item the hash $h(S_{k-1})$ of this subject,
        \item on $\left(H_{k-1} || h(S_{k^-2}) || h(S_{k-1})\right)$ for the \BCS of $c_{k-1}$,
        \item on $\left(H_{k-2} || h(S_{k-2}) || h(S_{k-1})\right)$ for the \FCS of~$c_{k-2}$,
    \end{itemize}
    where we recall that $H_{k-2}$ is the \Head of $c_{k- 2}$, $H_{k-1}$ is the \Head of $c_{k- 1}$, $h(S_{k-2})$ is the hash of the \Iss of $c_{k- 1}$.
    Then $\Si$ builds the $k$-th certificate, choosing a $H_k$ and $S_k$, using again the $\Key$ algorithm to sign $S_k$, querying the DSS challenger for two valid signatures, respectively, on:
    \begin{itemize}
        \item $\left(H_{k} || h(S_{k-1}) || h(S_{k})\right)$ for the \BCS of $c_{k}$,
        \item and on $\left(H_{k-1} || h(S_{k-1}) || h(S_{k})\right)$ for the \FCS of $c_{k-1}$,
    \end{itemize}
    where we recall that $H_{k}$ is the \Head of $c_{k}$ and $h(S_{k})$ is the hash of the \Sub of $c_{k}$.
    Finally, $\Si$ completes the chain $\mathrm{CC}$ (following the protocol and choosing everything, including the $\SK_i$'s), so that it has $n$ certificates, and passes it to~$\A$.

    $\A$ responds with a counterfeit chain $\mathrm{CC}'$ of length $i+1 \leq n$.
    $\A$ fails if and only if $\mathrm{CC}'$ is not valid, which happens when the last \Sub is not altered ($S_{i}'= S_{i}$) or when the integrity check of the verifier fails. If we are in this situation, $\Si$  discards $\mathrm{CC}'$ and gives up ($\Si$ fails).\\
    Otherwise, for $1\leq j \leq i$, $\Si$ controls the forged certificates $c_j' \in \mathrm{CC}'$ and compares them to $c_j \in \mathrm{CC}$. There are three cases to consider:
    \begin{itemize}
    \item for all $1\leq j \leq i$, $h(S_j)=h(S_j')$.\\
          Since $S_i'\ne S_i$, $\Si$ outputs the pair $(S_i, S_i')$ as a solution to the collision problem.
    \item there is $1\leq j \leq i$ such that $h(S_j) \neq h(S_j')$, but $h(S_k) = h(S_k')$ or $k>i$. \\
    In this case $\Si$ gives up ($\Si$ fails).
    \item $h(S_{k-1})=h(S'_{k-1})$ and $h(S_k) \ne h(S_k')$.\\
    If $S_{k-1} \ne S_{k-1}'$, the $\Si$ wins by submitting the pair $(S_{k-1}, S_{k-1}')$ as a solution to the collision problem. Otherwise, $S_{k-1} = S_{k-1}'$ and $PK^*$ is the public key of the issuer of $c_k'$.
    Then in the \BCS of the certificate $c_k'$ there is the digital signature $s$ for which holds the relation
     $\Ver\left(\left(H_{k}' || h(S_{k-1}') || h(S_{k}')\right), s, \PK^* \right) = \T$ (remember that $\mathrm{CC}'$ is considered valid, so the signatures check out), so $\Si$ may submit 
    $$
       \left(  \big( H_{k}' || h(S_{k-1}') || h(S_{k}') \big), s \right)
    $$
    as a winning answer of the challenge phase of the DSS game, since it is different from the messages $\Si$ queried for signatures,
    that are $$
    \left[\; h(S_{k-1}),\; \left( H_{k-1} || h(S_{k-2}) || h(S_{k-1}) \right) \right.,\; \left( H_{k-2} || h(S_{k-2}) || h(S_{k-1}) \right),$$
    $$ \left( H_{k} || h(S_{k-1}) || h(S_{k}) \right),\; \left. \left( H_{k-1} || h(S_{k-1}) || h(S_{k}) \right) \; \right]\,.
    $$ 
    \item $h(S_{k-1})\ne h(S'_{k-1})$ and $h(S_k) \ne h(S_k')$.\\
    $\Si$ gives up and fails.
\end{itemize}
    So if $\Si$ does not fail, $\Si$ correctly solves the collision problem or wins the DSS game, and since $\A$ is a polynomial-time algorithm, $\Si$ is a polynomial-time algorithm too.\\
    $\Si$ wins at least in the event when $\A$ wins and $S_k\ne S_k'$ (knowing that at least
    one $j$ exists such that $1\leq j\leq i\leq n-1$ $S_j \ne S_j'$). The probability of this
    event is at least the probability of the two cases and so it is
    $$
       \epsilon \cdot \frac{1}{n-1} \; = \; \frac{\epsilon}{n-1} \,.
    $$
  \end{proof}

  \begin{corollary}[STS Security]\label{stscor}
    If the Digital Signature Scheme is secure (see Assumption \ref{digsigsec}) and the hash function is collision resistant for the class $L$ (Assumption \ref{collres})  where $L$ is the set of all possible \Sub fields, then BIX protocol is secure against the static tampering with subject.
  \end{corollary}

  \begin{proof}
    For the BIX protocol to be functional the length of the chain must be polynomial, so for the result of Theorem~\ref{ststhm}, given a polynomial-time adversary that wins the STS game with non-negligible probability $\epsilon$, a polynomial-time simulator might be built that with probability at least $\frac{\epsilon}{n}$ either solves the collision problem \ref{collpb} or wins the digital signature security game \ref{digsiggame}, where $n$ is the length of the chain.
    But $\frac{\epsilon}{n}$ is non-negligible too, so this breaks either Assumption \ref{collres} or Assumption \ref{digsigsec}.
  \end{proof}

\section{Mid-Chain Altering}
\label{mid-alt}

Our proofs of security in the previous two sections does not show the impregnability of BIX to all
protocol attacks. In this section we do present an effective attack, where two non-consecutive members of the BIX community (i.e. whose BIX certificates are not next to each other in the chain) collude to create an alternate version of the chain between the two that is  considered valid by the members outside of that section, but that contains  subjects chosen arbitrarily by the two malicious users.

  Let $S_i, S_j$ be two malicious colluding users, where the indexes $i, j$ of their certificates in the certificate chain are  such that $j > i + 1 > i > 0$.
  Suppose that the chain is built properly up to the $j$-th certificate.
  We claim that, once $S_j$ has received his certificate $c_j$, he may collude with $S_i$ in order to change the information in the certificates $c_k$ with $i< k \leq j$ in such a way that every user $S_m$ with a certificate with index $0 < m < i$ or $m > j$ will consider correct the altered certificates (if they have not already obtained the original certificates).

  The first thing they do is to change the information in the  \Sub fields $S_j$ ($i+1 \leq j\leq j-1$) by generating  private keys and the corresponding  public ones (and then they are able to sign everything).
  Then, the first malicious user $S_i$ changes her certificate $c_i$ so that the fields \Next and \FCS link to the altered information and validate it, and similarly does $S_j$ with his fields \Iss and \BCS of $c_j$.

  At this point this altered version of the chain is considered valid by unsuspecting users. Moreover, $S_j$ as last user is responsible to supply the certificates in the chain to new users, so he may propagate the altered version, while older users $S_m$ with $0 < m < i$ will unwillingly authenticate altered certificates. Indeed,  when checking the integrity by traversing the chain either forward or backward, they find no inconsistency as long as $S_i$ points to the altered version.

\section{Conclusions}
\label{concl}
In this paper the BIX certificates protocol proposed in \cite{BIX} has been formally analyzed from a security point of view.
In particular the security against static attacks that aim to corrupt a chain has been proven, reducing the security to the choice of an adequate hash function and digital signature scheme.
For this reason the security of ECDSA, the main DSS nowadays, has also been discussed.\\
Moreover, an attack has been proposed on the BIX protocol in which two malicious users can collude to forge a portion of the certificate chain, suggesting that although the protocols has some solid security features, it still needs improvements to resist to more active attacks.

\bibliographystyle{plain}
\bibliography{biblio.bib%
}

\begin{thebibliography}{10}

\bibitem{PoE}
M.~Araoz.
\newblock Proof of existence.
\newblock Technical report, 2015.

\bibitem{X.509}
D.~Cooper and et~al.
\newblock Internet x.509 public key infrastructure certificate and certificate
  revocation list (crl) profile.
\newblock Technical report, IETF RFC 5280, 2008.

\bibitem{CGC-cry-art-elgamal}
T.~Elgamal.
\newblock {A public key cryptosystem and a signature scheme based on discret e
  logaritmhs}.
\newblock {\em IEEE Trans. on Inf. Th.}, 31(4):469--472, 1985.

\bibitem{ethereum}
Ethereum Foundation.
\newblock Ethereum project.
\newblock Technical report, 2015.

\bibitem{ECDSA}
D.~Johnson, A.~Menezes, and S.~Vanstone.
\newblock Technical report.

\bibitem{BIX}
Sead Muftic.
\newblock Bix certificates: Cryptographic tokens for anonymous transactions
  based on certificates public ledger.
\newblock {\em LEDGER}, 1(1):1--12, 2016.

\bibitem{sha}
NIST.
\newblock Secure hash standard (shs).
\newblock Technical Report FIPS PUB 180, 2012.

\bibitem{signatures-paillier}
P.~Paillier and D.~Vergnaud.
\newblock Discrete-log-based signatures may not be equivalent to discrete log.
\newblock {\em LNCS}, 3788:11--20, 2005.

\bibitem{CGC2-cry-art-Preneel99}
B.~Preneel.
\newblock {The State of Cryptographic Hash Functions}.
\newblock {\em LNCS}, 1561:158--182, 1999.

\bibitem{CGC-cry-art-Rab79}
M.~O. Rabin.
\newblock {Digital signatures and public-key functions as intractable as
  factorization}.
\newblock Technical Report MIT/LCS/TR-212, MIT laboratory for computer science,
  01 1979.

\bibitem{CGC-cry-art-RSA78}
R.~L. Rivest, A.~Shamir, and L.~M. Adleman.
\newblock {A Method for Obtaining Digital Signatures and Public-Key
  Cryptosystems}.
\newblock {\em Commun. ACM}, 21(2):120--126, 1978.

\bibitem{storj}
S.~Wilkinson, T.~Boshevski, J.~Brandoff, and V.~Buterin.
\newblock Storj: A peer-to-peer cloud storage network.
\newblock Technical report, 2014.

\end{thebibliography}

\end{document}